\documentclass[11pt]{article}

\usepackage[a4paper,margin=1in]{geometry}
\usepackage{amsmath,amssymb,amsthm,mathtools}
\usepackage{algorithm}
\usepackage{algpseudocode}
\usepackage{enumitem}
\usepackage{hyperref}
\usepackage{xcolor}
\usepackage{graphicx}
\usepackage{booktabs}
\usepackage{tikz}
\usetikzlibrary{arrows.meta,positioning,calc,fit,backgrounds,decorations.pathreplacing}
\hypersetup{
  colorlinks=true,
  linkcolor=blue!60!black,
  citecolor=blue!60!black,
  urlcolor=blue!60!black
}

\newtheorem{theorem}{Theorem}[section]
\newtheorem{lemma}[theorem]{Lemma}
\newtheorem{corollary}[theorem]{Corollary}

\newtheorem{definition}[theorem]{Definition}

\newcommand{\dist}{\operatorname{dist}}
\newcommand{\rad}{\operatorname{rad}}
\newcommand{\diam}{\operatorname{diam}}
\newcommand{\ecc}{\operatorname{ecc}}
\newcommand{\Ball}{\operatorname{B}}
\newcommand{\cost}{\operatorname{cost}}
\newcommand{\PathCost}{\gamma_{\mathrm{path}}}
\newcommand{\bd}{\gamma_b}
\newcommand{\emptyword}{\varnothing}
\newcommand{\Req}{\operatorname{Req}}
\newcommand{\Comp}{\operatorname{Comp}}

\title{An $O(n^5)$-Time Algorithm for Optimal Broadcast Domination}
\author{Kleitos Papadopoulos}
\date{\today}

\begin{document}
\maketitle

\begin{abstract}
Broadcast domination assigns a nonnegative integer power to every vertex of a graph so that every vertex is within the assigned power of some broadcasting vertex, and the objective is to minimize the sum of the powers. Heggernes and Lokshtanov proved that the problem is polynomial-time solvable on arbitrary connected unweighted graphs by showing that some optimal efficient broadcast has a domination graph that is a path or a cycle, and by reducing the general case to an $O(n^6)$-time algorithm.

This paper gives an efficient algorithm of the path-case. Instead of building one auxiliary acyclic graph for every possible left endpoint vertex, we build a single directed acyclic graph whose states are oriented broadcast balls together with their two possible residual sides. The resulting path-case algorithm runs in $O(n^3)$ time and $O(n^3)$ space on an $n$-vertex graph. Combining this routine with the same peel-one-ball reduction of Heggernes and Lokshtanov yields an exact $O(n^5)$-time algorithm for optimal broadcast domination on arbitrary connected unweighted graphs. This resolves the quintic-time conjecture for general graphs attributed to Heggernes and S\ae ther and recorded in subsequent surveys of broadcast domination.
\end{abstract}

\section{Introduction}

Let $G=(V,E)$ be a connected, simple, unweighted graph with $|V|=n$. A \emph{broadcast} is a function
\[
  f:V\to \{0,1,\ldots,\diam(G)\}.
\]
A vertex $v$ with $f(v)>0$ broadcasts with power $f(v)$ and dominates every vertex in the ball
\[
  \Ball_G(v,f(v))=\{x\in V: \dist_G(v,x)\le f(v)\}.
\]
The broadcast is dominating if every vertex is dominated by at least one broadcasting vertex. Its cost is
\[
  \cost(f)=\sum_{v\in V} f(v).
\]
The optimum value is denoted $\bd(G)$. Throughout the main text we assume that all input graphs have at least two vertices. The one-vertex graph is handled separately by the convention $\bd(K_1)=0$.

Broadcast domination was introduced by Erwin~\cite{Erwin04} and further developed by Dunbar et al.~\cite{Dunbar06}. Heggernes and Lokshtanov proved that optimal broadcast domination is polynomial-time solvable on arbitrary connected graphs \cite{HL06}. Their central structural result says that every graph has an optimal efficient broadcast whose domination graph is either a path or a cycle. They then solve the path case by constructing, for every possible anchor vertex $u$, an auxiliary acyclic digraph $G_u$. Each $G_u$ has $O(n^2)$ vertices and $O(n^3)$ arcs; after preprocessing, the path case is solved in $O(n^4)$ time. The full algorithm guesses one broadcast ball to delete and calls the path routine on the residual graph, giving $O(n^6)$ time \cite{HL06}. Subsequent work treated this sextic bound as the best known general-graph running time: Heggernes and S\ae ther gave linear time for block graphs and an $O(n^4)$ consequence for chordal graphs \cite{HS12}, while later work still noted that no faster algorithm was known for general undirected graphs \cite{FoucaudGPS21}. The survey chapter of Henning, MacGillivray, and Yang records the conjecture, attributed to Heggernes and S\ae ther, that \textsc{Broadcast Domination} can be solved in $O(n^5)$ time in general \cite{HMY21}. Thus the main theorem below resolves this quintic-time conjecture.

The improvement here is to remove the anchor loop from the path case. The anchor $u$ serves only to decide which residual component of a candidate ball lies to the left and which lies to the right. We encode this orientation directly. A state is an oriented ball
\[
  (v,p,L,R),
\]
where $L$ and $R$ are the left and right residual sides of $G-\Ball(v,p)$. Arcs express that two oriented balls may be consecutive in a path-shaped broadcast. These arcs are acyclic because the left side strictly grows along every arc. Since adjacent efficient balls are distance-tight, the number of possible arcs is $O(n^3)$.

The main results are as follows.

\begin{theorem}
Let $G$ be a connected unweighted graph on $n$ vertices. The minimum-cost efficient broadcast domination whose domination graph is a path can be computed in $O(n^3)$ time and $O(n^3)$ space.
\end{theorem}

\begin{theorem}
\label{thm:general-main}
Let $G$ be a connected unweighted graph on $n$ vertices. An optimal broadcast domination of $G$ can be computed in $O(n^5)$ time and $O(n^3)$ working space.
\end{theorem}

Consequently, Theorem~\ref{thm:general-main} proves the general-graph quintic bound conjectured by Heggernes and S\ae ther and cited in later surveys \cite{HMY21}.

\section{Preliminaries}

For $S\subseteq V(G)$, let $G[S]$ be the subgraph induced by $S$, and let
\[
  N_G(S)=\{x\in V(G)\setminus S: xy\in E(G)\text{ for some }y\in S\}
\]
be the open neighborhood of $S$. We omit the subscript $G$ when the graph is clear.

A broadcast $f$ is \emph{efficient} if the balls of distinct broadcasting vertices are pairwise disjoint. Let
\[
  V_f=\{v\in V(G): f(v)>0\}.
\]
For an efficient broadcast $f$, its \emph{domination graph} $G_f$ has vertex set $V_f$ and an edge $uv$ whenever there is an edge of $G$ between $\Ball(u,f(u))$ and $\Ball(v,f(v))$. Equivalently, $G_f$ is obtained by contracting every active broadcast ball to one weighted vertex and preserving inter-ball adjacencies.

\begin{lemma}[Dunbar et al.~{\cite{Dunbar06}}]
\label{lem:efficient-optimum}
Every connected graph $G$ has an efficient optimal broadcast domination.
Equivalently, there is a dominating broadcast $f$ with
\[
  \cost(f)=\bd(G)
\]
such that the balls
\[
  \{\Ball_G(v,f(v)) : v\in V_f\}
\]
are pairwise disjoint.
\end{lemma}

\begin{proof}
This is a theorem of Dunbar et al.~\cite{Dunbar06}. In particular, their argument shows that if $f$ is a non-efficient dominating broadcast, then one can replace $f$ by an efficient dominating broadcast $f'$ with the same cost. Applying this replacement to an optimal dominating broadcast gives an efficient optimal broadcast domination of $G$.
\end{proof}

We use the following stronger structure theorem.

\begin{theorem}[Heggernes--Lokshtanov~\cite{HL06}]
\label{thm:HLstructure}
Every connected graph has an efficient optimal broadcast domination $f$ such that $G_f$ is a path or a cycle.
\end{theorem}

The proof of Theorem~\ref{thm:HLstructure} uses two facts. First, every non-efficient optimum can be converted into an efficient optimum of the same cost with fewer active vertices. Second, if an efficient broadcast has a domination graph with a vertex of degree greater than two, one can replace several neighboring broadcasts by one broadcast of the same total cost, again reducing the number of active vertices. Repeating this gives a connected domination graph of maximum degree at most two.

For a connected graph $H$, define $\PathCost(H)$ to be the minimum cost of an efficient dominating broadcast $f$ of $H$ such that $H_f$ is a path. If a single broadcast ball dominates all of $H$, we regard the domination graph as a one-vertex path. For the one-vertex graph we use the same convention as above and set $\PathCost(K_1)=0$.

The full algorithm also uses the following peel lemma of Heggernes and Lokshtanov.

\begin{corollary}[Heggernes--Lokshtanov~{\cite[Corollary 3.5]{HL06}}]
\label{cor:peel}
For every connected graph $G$, there is an efficient optimal broadcast $f$ and an active vertex $x\in V_f$ such that, for $k=f(x)$, the graph
\[
  G'=G[V(G)\setminus \Ball_G(x,k)]
\]
is either empty, or is connected and satisfies
\[
  \bd(G')=\PathCost(G').
\]
Equivalently, when $G'$ is nonempty, it has an optimal efficient broadcast whose domination graph is a path.
\end{corollary}

The convention $\PathCost(K_1)=0$ is useful for the standalone path problem, but it must be treated carefully in the outer peel algorithm. A residual singleton is not already dominated by the peeled ball, so the general algorithm below handles singleton residuals explicitly rather than invoking the path routine on them.

\section{Metric facts about efficient balls}

We first state two elementary facts used throughout the algorithm.

\begin{lemma}[Ball intersection]
\label{lem:intersection}
For vertices $a,b$ and integers $p,q\ge 0$,
\[
  \Ball(a,p)\cap \Ball(b,q)\ne\emptyword
  \quad\Longleftrightarrow\quad
  \dist(a,b)\le p+q.
\]
\end{lemma}

\begin{proof}
If $z$ lies in both balls, then $\dist(a,b)\le \dist(a,z)+\dist(z,b)\le p+q$. Conversely, if $\dist(a,b)\le p+q$, take a shortest $a$-$b$ path and choose a vertex $z$ on it at distance at most $p$ from $a$ and at distance at most $q$ from $b$.
\end{proof}

\begin{lemma}
\label{lem:tight}
Let $X=\Ball(a,p)$ and $Y=\Ball(b,q)$ be disjoint balls. If there is an edge of $G$ between $X$ and $Y$, then
\[
  \dist(a,b)=p+q+1.
\]
Conversely, if $\dist(a,b)=p+q+1$, then there is an edge between $X$ and $Y$.
\end{lemma}

\begin{proof}
Since $X$ and $Y$ are disjoint, Lemma~\ref{lem:intersection} gives $\dist(a,b)>p+q$. If there is an edge $xy$ with $x\in X$ and $y\in Y$, then
\[
  \dist(a,b)\le \dist(a,x)+1+\dist(y,b)\le p+q+1.
\]
Thus equality holds. Conversely, if $\dist(a,b)=p+q+1$, then on a shortest $a$-$b$ path, the edge between the vertex at distance $p$ from $a$ and the next vertex is an edge from $X$ to $Y$.
\end{proof}

\section{An anchor-free path-case algorithm}
\label{sec:path}

Let $H=(V,E)$ be a connected graph on $s$ vertices. In this section we compute $\PathCost(H)$ in $O(s^3)$ time. If $s=1$, we return the zero broadcast by the convention $\PathCost(K_1)=0$. Therefore assume for the remainder of this section that $s\ge 2$.

Throughout this section all distances, balls, and neighborhoods are taken in $H$ unless explicitly stated otherwise. Let $\rho=\rad(H)$. Since $s\ge 2$, we have $\rho\ge 1$. A central broadcast of radius $\rho$ is a feasible path-shaped broadcast, since its domination graph has one vertex. Hence $\PathCost(H)\le \rho$. Therefore no minimum-cost path-shaped broadcast uses a power greater than $\rho$, since such a single power would already make the total cost exceed $\rho$. Thus it suffices to consider balls $(v,p)$ with $v\in V$ and $1\le p\le \rho$.

\subsection{Oriented balls}

For a ball $b=(v,p)$, write
\[
  X_b=\Ball(v,p),
  \qquad
  R_b^0=V\setminus X_b.
\]
Let $\kappa(b)$ be the number of connected components of $H[R_b^0]$. We only keep balls with $\kappa(b)\le 2$. Intuitively, because the target domination graph forms a path, deleting any internal broadcast ball can partition the remaining active balls into at most two directions---left and right. Any third residual component would contain no active balls and thus remain completely undominated.

\begin{definition}[State]
A \emph{state} is a tuple
\[
  \sigma=(b,L_\sigma,R_\sigma),
  \qquad b=(v,p),
\]
constructed as follows.
\begin{itemize}[leftmargin=2em]
  \item If $H[V\setminus X_b]=\emptyword$, create one radial state with $L_\sigma=R_\sigma=\emptyword$.
  \item If $H[V\setminus X_b]$ is connected with component $C$, create two endpoint states $(b,\emptyword,C)$ and $(b,C,\emptyword)$.
  \item If $H[V\setminus X_b]$ has exactly two connected components $C_1,C_2$, create two internal states $(b,C_1,C_2)$ and $(b,C_2,C_1)$.
\end{itemize}
The ball, center, and radius of a state $\sigma$ are denoted $X_\sigma$, $c_\sigma$, and $r_\sigma$.
\end{definition}

The intended interpretation is that $L_\sigma$ is the part of the path already placed to the left of $X_\sigma$, and $R_\sigma$ is the part that remains to the right.

For a state $\sigma$, define its left and right frontiers by
\[
  F^L_\sigma=N(X_\sigma)\cap L_\sigma,
  \qquad
  F^R_\sigma=N(X_\sigma)\cap R_\sigma.
\]

\begin{figure}[t]
\centering
\begin{tikzpicture}[
  every node/.style={font=\small},
  component/.style={
    draw,
    rounded corners=22pt,
    fill=gray!10,
    minimum width=2.45cm,
    minimum height=1.45cm
  },
  ball/.style={
    draw,
    circle,
    fill=blue!10,
    minimum size=1.55cm
  },
  frontier/.style={
    draw=orange!85!black,
    fill=orange!28,
    rounded corners=4pt,
    minimum width=.16cm,
    minimum height=.9cm
  },
  center/.style={circle, fill=black, inner sep=1.15pt},
  leader/.style={-{Latex[length=1.8mm]}, thin},
  contact/.style={line width=.7pt, densely dashed},
  labelnode/.style={
    font=\footnotesize,
    align=center,
    fill=white,
    inner sep=1.5pt
  },
  note/.style={
    font=\scriptsize,
    align=center,
    fill=white,
    inner sep=1.5pt
  }
]

  \node[component] (L) at (-3.15,0) {};
  \node[ball]      (X) at (0,0) {};
  \node[component] (R) at (3.15,0) {};

  \node[frontier] (FL) at (-2.05,0) {};
  \node[frontier] (FR) at (2.05,0) {};

  \node[center] (c) at (0,0) {};

  \draw[contact] (FL.east) -- (X.west);
  \draw[contact] (X.east) -- (FR.west);

  \node[labelnode] (labL) at (-3.15,1.35) {$L_\sigma$};
  \node[labelnode] (labX) at (0,1.35) {$X_\sigma=\Ball(c_\sigma,r_\sigma)$};
  \node[labelnode] (labR) at (3.15,1.35) {$R_\sigma$};

  \draw[leader] (labL.south) -- (L.north);
  \draw[leader] (labX.south) -- (X.north);
  \draw[leader] (labR.south) -- (R.north);

  \node[labelnode] (labc) at (0,-1.15) {$c_\sigma$};
  \draw[leader] (labc.north) -- (c.south);

  \node[labelnode] (labFL) at (-4.45,-1.05)
    {$F^L_\sigma$\\[-1pt]$=N(X_\sigma)\cap L_\sigma$};

  \node[labelnode] (labFR) at (4.45,-1.05)
    {$F^R_\sigma$\\[-1pt]$=N(X_\sigma)\cap R_\sigma$};

  \draw[leader] (labFL.east) -- (FL.south west);
  \draw[leader] (labFR.west) -- (FR.south east);

  \begin{scope}[on background layer]
    \node[
      draw,
      dashed,
      rounded corners=24pt,
      fit=(L)(X)(R)(labL)(labX)(labR),
      inner sep=.3cm,
      label={[font=\small]above:$H$}
    ] {};
  \end{scope}

  \node[note] at (0,-1.9)
    {Orange bands denote frontier sets};

\end{tikzpicture}
\caption{An oriented state $\sigma=(b,L_\sigma,R_\sigma)$, where $b=(c_\sigma,r_\sigma)$ and $X_\sigma=\Ball(c_\sigma,r_\sigma)$. The sets $L_\sigma$ and $R_\sigma$ are the two oriented residual sides of $H[V(H)\setminus X_\sigma]$.}
\label{fig:oriented-state}
\end{figure}
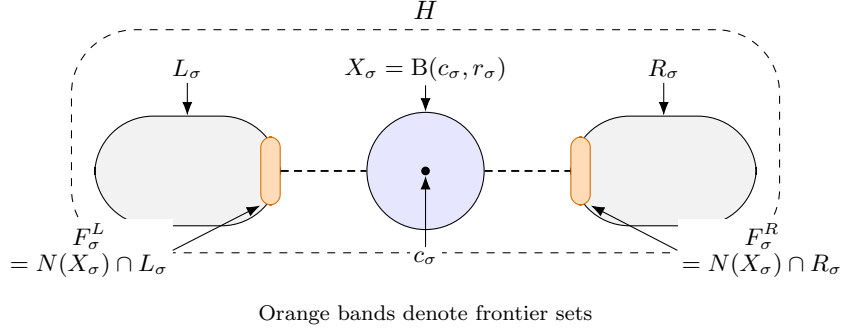

\subsection{The state digraph}

We build a weighted directed graph $D(H)$. Its vertices are all states. The weight of a state $\sigma$ is $r_\sigma$.

There is an arc $\sigma\to\tau$ if and only if all of the following conditions hold:
\begin{enumerate}[label=(A\arabic*),leftmargin=3em]
  \item $R_\sigma\ne\emptyword$ and $L_\tau\ne\emptyword$.
  \item The two balls are distance-tight:
  \[
    \dist(c_\sigma,c_\tau)=r_\sigma+r_\tau+1.
  \]
  \item $c_\tau\in R_\sigma$ and $c_\sigma\in L_\tau$.
  \item The exposed right frontier of $\sigma$ is covered by $\tau$, and the exposed left frontier of $\tau$ is covered by $\sigma$:
  \[
    F^R_\sigma\subseteq X_\tau,
    \qquad
    F^L_\tau\subseteq X_\sigma.
  \]
\end{enumerate}

\begin{figure}[t]
\centering
\begin{tikzpicture}[
  every node/.style={font=\small},
  component/.style={
    draw,
    rounded corners=18pt,
    fill=gray!10,
    minimum width=1.85cm,
    minimum height=1.2cm
  },
  largecomponent/.style={
    draw,
    rounded corners=22pt,
    fill=gray!7,
    minimum width=5.8cm,
    minimum height=1.45cm
  },
  ball/.style={
    draw,
    circle,
    fill=blue!10,
    minimum size=1.45cm
  },
  frontier/.style={
    draw=orange!85!black,
    fill=orange!30,
    rounded corners=4pt,
    minimum width=.15cm,
    minimum height=.72cm
  },
  center/.style={circle, fill=black, inner sep=1.05pt},
  leader/.style={-{Latex[length=1.8mm]}, thin},
  arrow/.style={-{Latex[length=2.3mm]}, thick},
  contact/.style={line width=.7pt, densely dashed},
  labelnode/.style={
    font=\footnotesize,
    align=center,
    fill=white,
    inner sep=1.5pt
  },
  rowlabel/.style={
    font=\footnotesize,
    align=right,
    fill=white,
    inner sep=1.5pt
  },
  formula/.style={
    font=\footnotesize,
    align=center,
    fill=white,
    inner sep=2pt
  }
]

  \def\ytop{1.55}
  \def\ybot{-1.65}

  \node[component]      (LsT) at (-5.40,\ytop) {};
  \node[ball]           (XsT) at (-3.15,\ytop) {};
  \node[largecomponent] (RsT) at (1.45,\ytop) {};
  \node[ball]           (XtT) at (-0.55,\ytop) {};
  \node[component]      (RtT) at (2.75,\ytop) {};

  \node[frontier] (FRs) at (-1.08,\ytop) {};
  \draw[contact] (XsT.east) -- (FRs.west);

  \node[center] (csT) at (-3.15,\ytop) {};
  \node[center] (ctT) at (-0.55,\ytop) {};

  \node[rowlabel] at (-7.05,\ytop)
    {view from $\sigma$:\\delete $X_\sigma$};

  \node[labelnode] (labXsT) at (-3.15,2.72) {$X_\sigma$};
  \node[labelnode] (labXtT) at (-0.55,2.72) {$X_\tau\subseteq R_\sigma$};
  \node[labelnode] (labRsT) at (3.65,2.48) {$R_\sigma$};
  \node[labelnode] (labFRs) at (-0.15,.35)
    {$F^R_\sigma=N(X_\sigma)\cap R_\sigma$\\[-1pt]$\subseteq X_\tau$};

  \draw[leader] (labXsT.south) -- (XsT.north);
  \draw[leader] (labXtT.south) -- (XtT.north);
  \draw[leader] (labRsT.south west) -- (RsT.north east);
  \draw[leader] (labFRs.north west) -- (FRs.south);

  \node[largecomponent] (LtB) at (-3.55,\ybot) {};
  \node[component]      (LsB) at (-5.40,\ybot) {};
  \node[ball]           (XsB) at (-3.15,\ybot) {};
  \node[ball]           (XtB) at (-0.55,\ybot) {};
  \node[component]      (RtB) at (2.75,\ybot) {};

  \node[frontier] (FLt) at (-2.62,\ybot) {};
  \draw[contact] (FLt.east) -- (XtB.west);

  \node[rowlabel] at (-7.05,\ybot)
    {view from $\tau$:\\delete $X_\tau$};

  \node[labelnode] (labLtB) at (-5.05,-.45) {$L_\tau$};
  \node[labelnode] (labXsB) at (-2.92,-.45) {$X_\sigma\subseteq L_\tau$};
  \node[labelnode] (labXtB) at (-0.55,-.45) {$X_\tau$};
  \node[labelnode] (labFLt) at (-3.35,-3.05)
    {$F^L_\tau=N(X_\tau)\cap L_\tau$\\[-1pt]$\subseteq X_\sigma$};

  \draw[leader] (labLtB.south) -- (LtB.north west);
  \draw[leader] (labXsB.south) -- (XsB.north);
  \draw[leader] (labXtB.south) -- (XtB.north);
  \draw[leader] (labFLt.north) -- (FLt.south);

  \draw[arrow] (-3.15,3.25) -- node[above] {$\sigma\to\tau$} (-0.55,3.25);

  \node[formula] at (2.6,-.15)
    {$\dist(c_\sigma,c_\tau)=r_\sigma+r_\tau+1$};

  \node[formula] at (2.6,-.55)
    {contact strokes are schematic};

  \begin{scope}[on background layer]
    \node[
      draw,
      dashed,
      rounded corners=18pt,
      fit=(LsT)(XsT)(RsT)(labXsT)(labXtT)(labRsT),
      inner sep=.25cm
    ] {};

    \node[
      draw,
      dashed,
      rounded corners=18pt,
      fit=(LtB)(XtB)(RtB)(labFLt),
      inner sep=.25cm
    ] {};
  \end{scope}

\end{tikzpicture}
\caption{A directed arc $\sigma\to\tau$ shown through the two relevant residual decompositions. In the $\sigma$-decomposition, $X_\tau$ lies in the right residual side $R_\sigma$ and covers the exposed right frontier of $\sigma$. In the $\tau$-decomposition, $X_\sigma$ lies in the left residual side $L_\tau$ and covers the exposed left frontier of $\tau$.}
\label{fig:state-arc}
\end{figure}
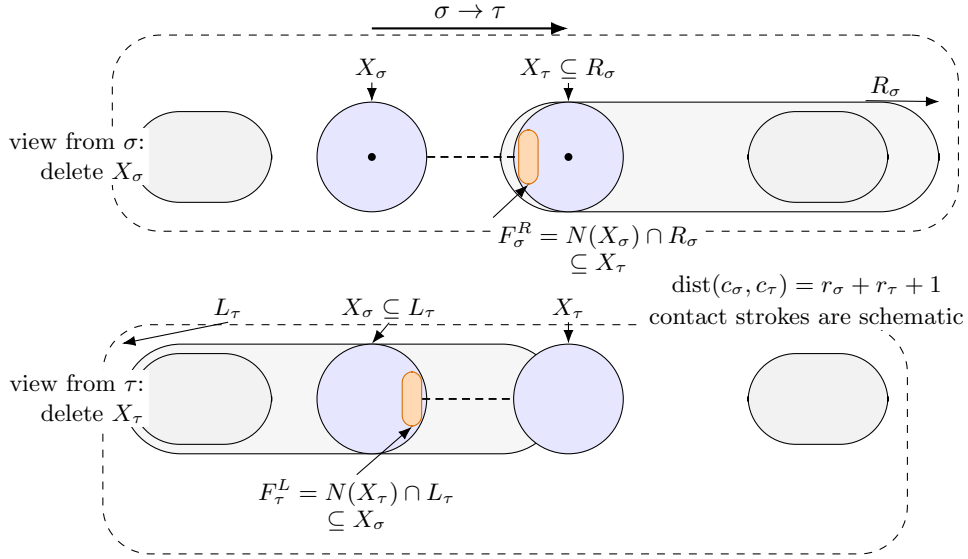

A state with $L_\sigma=\emptyword$ is a \emph{source state}. A state with $R_\sigma=\emptyword$ is a \emph{sink state}. A radial state is both a source and a sink.

\begin{lemma}
\label{lem:containment}
If $\sigma\to\tau$ is an arc of $D(H)$, then
\[
  X_\tau\subseteq R_\sigma,
  \qquad
  X_\sigma\subseteq L_\tau.
\]
\end{lemma}

\begin{proof}
By condition (A2), Lemma~\ref{lem:intersection} gives $X_\sigma\cap X_\tau=\emptyword$. The ball $X_\tau$ is connected and contains $c_\tau\in R_\sigma$. Since $R_\sigma$ is a connected component of $H[V\setminus X_\sigma]$, the whole ball $X_\tau$ lies in $R_\sigma$. The proof of $X_\sigma\subseteq L_\tau$ is symmetric.
\end{proof}

\begin{lemma}
\label{lem:dag}
The digraph $D(H)$ is acyclic. More precisely, every arc $\sigma\to\tau$ satisfies
\[
  |L_\tau|>|L_\sigma|.
\]
\end{lemma}

\begin{proof}
By Lemma~\ref{lem:containment}, $X_\sigma\subseteq L_\tau$ and $X_\tau\subseteq R_\sigma$. The set $L_\sigma\cup X_\sigma$ is connected when $L_\sigma$ is nonempty, and the claim is trivial when $L_\sigma$ is empty. Moreover, $L_\sigma\cup X_\sigma$ avoids $X_\tau$. Since $L_\tau$ is the component of $H[V\setminus X_\tau]$ containing $X_\sigma$, we have
\[
  L_\sigma\cup X_\sigma\subseteq L_\tau.
\]
As $X_\sigma$ is nonempty and disjoint from $L_\sigma$, this implies $|L_\tau|>|L_\sigma|$.
\end{proof}

\subsection{Paths in the state digraph are broadcasts}

Let
\[
  P=\sigma_1,\sigma_2,\ldots,\sigma_t
\]
be a directed path in $D(H)$ from a source state to a sink state. Define a broadcast $f_P$ by
\[
  f_P(c_{\sigma_i})=r_{\sigma_i}\quad (1\le i\le t),
\]
and by assigning zero to all other vertices. If the same center appears twice, keep the larger assigned value; however, by Lemma~\ref{lem:dag} and Lemma~\ref{lem:containment}, this never happens on a simple directed path.

\begin{lemma}
\label{lem:soundness}
For every directed path $P$ in $D(H)$ from a source state to a sink state, $f_P$ is an efficient dominating broadcast of $H$, and $(H)_{f_P}$ is a path. Its cost is the sum of the weights of the states in $P$.
\end{lemma}

\begin{proof}
Write $P=\sigma_1,\ldots,\sigma_t$ and $X_i=X_{\sigma_i}$. First note that the proof of Lemma~\ref{lem:dag} gives the following stronger statement for each arc $\sigma_i\to\sigma_{i+1}$:
\[
  L_{\sigma_i}\cup X_i\subseteq L_{\sigma_{i+1}}.
\]
Consequently, if $i<j$, then $X_i\subseteq L_{\sigma_j}$. Since $X_j$ is disjoint from $L_{\sigma_j}$ by definition of a state, the balls $X_1,\ldots,X_t$ are pairwise disjoint.

We next record the rightward and leftward containment of later and earlier balls. If $i<j$, then $X_j\subseteq R_{\sigma_i}$; symmetrically, if $j<i$, then $X_j\subseteq L_{\sigma_i}$. For $j=i+1$ this is Lemma~\ref{lem:containment}. If $X_j\subseteq R_{\sigma_i}$ and $j<t$, then $X_j$ is adjacent to $X_{j+1}$ by condition (A2) and Lemma~\ref{lem:tight}, and $X_{j+1}$ is disjoint from $X_i$ by the preceding paragraph. Since $R_{\sigma_i}$ is a connected component of $H[V\setminus X_i]$, the connected set $X_{j+1}$ must also lie in $R_{\sigma_i}$. Induction proves the rightward claim, and the leftward claim is symmetric.

Thus the broadcast $f_P$ is efficient once we prove that it dominates $H$.

Let
\[
  S=\bigcup_{i=1}^t X_i.
\]
We prove $S=V$. Suppose not, and choose a vertex $z\in V\setminus S$. Since $H$ is connected, there is a path from $z$ to $S$. Let $y$ be the last vertex outside $S$ on such a path and let $x$ be the next vertex, so $x\in S$ and $y\in N(S)$. Choose $i$ with $x\in X_i$. Then $y\in N(X_i)$ and $y\notin X_i$.

If $t=1$, then the unique state is both a source and a sink, hence it is radial and $X_1=V$, a contradiction. Assume $t\ge 2$. If $i=1$, then $L_{\sigma_1}=\emptyword$, so $y\in F^R_{\sigma_1}$. The arc $\sigma_1\to\sigma_2$ gives $F^R_{\sigma_1}\subseteq X_2\subseteq S$, contradicting $y\notin S$. The case $i=t$ is symmetric. If $1<i<t$, then $y$ belongs either to $F^L_{\sigma_i}$ or to $F^R_{\sigma_i}$. The arc $\sigma_{i-1}\to\sigma_i$ gives $F^L_{\sigma_i}\subseteq X_{i-1}$, while the arc $\sigma_i\to\sigma_{i+1}$ gives $F^R_{\sigma_i}\subseteq X_{i+1}$. In either case $y\in S$, again a contradiction. Thus $S=V$.

It remains to verify that the domination graph is a path. Consecutive balls $X_i$ and $X_{i+1}$ are adjacent by condition (A2) and Lemma~\ref{lem:tight}. No nonconsecutive balls are adjacent. To see this, suppose $i+1<j$ and there is an edge between $X_i$ and $X_j$. By the monotonic containment above, $X_j\subseteq R_{\sigma_i}$, so the endpoint of this edge in $X_j$ lies in $F^R_{\sigma_i}$. But the arc $\sigma_i\to\sigma_{i+1}$ gives $F^R_{\sigma_i}\subseteq X_{i+1}$, contradicting the pairwise disjointness of $X_{i+1}$ and $X_j$. Therefore the only edges of the domination graph are between consecutive balls in $P$, so $(H)_{f_P}$ is a path.

The cost is exactly $\sum_i r_{\sigma_i}$ by the definition of $f_P$.
\end{proof}

\subsection{Every path broadcast appears}

\begin{lemma}
\label{lem:completeness}
Let $f$ be a minimum-cost efficient dominating broadcast of $H$ among those whose domination graph is a path. Then there is a source-to-sink directed path $P$ in $D(H)$ whose state weights sum to $\cost(f)$.
\end{lemma}

\begin{proof}
Let the active vertices of $f$ be ordered along the domination path as
\[
  a_1,a_2,\ldots,a_t,
\]
and let $p_i=f(a_i)$ and $X_i=\Ball(a_i,p_i)$. Since $f$ is efficient and dominating, the balls $X_1,\ldots,X_t$ are pairwise disjoint and their union is $V(H)$. If $t=1$, then $X_1=V(H)$ and the corresponding radial state is a source-to-sink path of weight $p_1$.

Assume $t\ge 2$. For $1\le i\le t$, define
\[
  L_i=\bigcup_{j<i} X_j,
  \qquad
  R_i=\bigcup_{j>i} X_j.
\]
When nonempty, $L_i$ is connected because consecutive balls $X_j,X_{j+1}$ with $j<i-1$ are adjacent in the domination graph, and the same argument shows that $R_i$ is connected. Moreover, no edge of $H$ joins $L_i$ to $R_i$, since such an edge would give an edge in $H_f$ between nonconsecutive vertices of the domination path. Because the balls partition $V(H)$, the graph $H[V(H)\setminus X_i]$ therefore has exactly the nonempty components among $L_i$ and $R_i$. Thus
\[
  \sigma_i=((a_i,p_i),L_i,R_i)
\]
is a state.

For every $i<t$, the balls $X_i$ and $X_{i+1}$ are disjoint and adjacent, so by Lemma~\ref{lem:tight},
\[
  \dist(a_i,a_{i+1})=p_i+p_{i+1}+1.
\]
Also, every vertex of $N(X_i)\cap R_i$ lies in $X_{i+1}$. Indeed, since the balls partition $V(H)$, such a vertex lies in some $X_j$ with $j>i$; the edge from $X_i$ to this vertex would give an edge from $X_i$ to $X_j$ in $H_f$, so the path order forces $j=i+1$. Hence
\[
  N(X_i)\cap R_i\subseteq X_{i+1}.
\]
The symmetric argument gives
\[
  N(X_{i+1})\cap L_{i+1}\subseteq X_i.
\]
Therefore $\sigma_i\to\sigma_{i+1}$ is an arc of $D(H)$. Since $L_1=\emptyword$ and $R_t=\emptyword$, the sequence $\sigma_1,\ldots,\sigma_t$ is a source-to-sink path of total weight $\cost(f)$.
\end{proof}

\subsection{Computing the state digraph in cubic time}

We now show that $D(H)$ can be built in $O(s^3)$ time.

First compute all-pairs shortest-path distances by running breadth-first search from every vertex. Since $H$ is unweighted and has at most $s^2$ edges, this takes $O(s^3)$ time.

For every center $v$, we compute the connected components of
\[
  H[V\setminus \Ball(v,p)]
\]
for all $p=1,\ldots,\rho$ in total $O(s^2)$ time for this fixed $v$. Process the radii in decreasing order. When moving from radius $p+1$ to radius $p$, add exactly the vertices at distance $p+1$ from $v$ and union them with already active neighbors. To avoid an inverse-Ackermann factor $\alpha(s)$ in the time complexity, we maintain the components using arrays of linked lists and perform union-by-size. Since we only merge components and do not require mid-loop \texttt{find} queries, the total cost of updating component identifiers is bounded by $O(s \log s)$ over all merges for this center. Each vertex and edge is considered only a constant number of times for this center. Scanning all vertices after each radius to record the component label of every active vertex costs $O(s)$ per radius, summing to $O(s^2)$ per center. Repeating over all centers costs $O(s^3)$.

For a ball $b=(v,p)$ and a residual component $C$ of $H[V\setminus X_b]$, define the requirement value
\[
  \Req(v,p,C,w)=
  \max\{\dist(w,z): z\in N(X_b)\cap C\},
\]
with maximum $0$ over the empty set. Then
\[
  N(X_b)\cap C\subseteq \Ball(w,q)
  \quad\Longleftrightarrow\quad
  \Req(v,p,C,w)\le q.
\]
We store requirement values only for residual components of balls $(v,p)$ with $\kappa(v,p)\le 2$, since only these balls generate states. By mapping the (at most two) residual components of each kept ball to a local binary index $c \in \{1, 2\}$, the requirement array $\Req(v,p,c,w)$ requires strictly $O(s^3)$ space. All such requirement values are computed in $O(s^3)$ time as follows. A vertex $z$ belongs to $N(\Ball(v,p))$ exactly when $\dist(v,z)=p+1$. Thus, for every ordered pair $(v,z)$, the relevant radius is uniquely determined as $p=\dist(v,z)-1$. If $1\le p\le\rho$ and $(v,p)$ is a kept ball, find the local component index $c$ of $z$ in $H[V\setminus \Ball(v,p)]$, and for every $w\in V$ update
\[
  \Req(v,p,c,w)\leftarrow
  \max\{\Req(v,p,c,w),\dist(w,z)\}.
\]
There are $O(s^2)$ choices of $(v,z)$ and $O(s)$ choices of $w$.

Finally, enumerate possible arcs. For every ordered triple $(v,w,p)$, set
\[
  q=\dist(v,w)-p-1.
\]
Only this value of $q$ can give a tight contact between $(v,p)$ and $(w,q)$. If $1\le q\le \rho$, inspect the constant number of states over $(v,p)$ and the constant number of states over $(w,q)$. Conditions (A1)--(A3) are checked from the stored component labels and distances; condition (A4) is checked by two requirement comparisons. Hence arc generation takes $O(s^3)$ time.

\begin{algorithm}[t]
\caption{\textsc{Fast-Path-Broadcast}$(H)$}
\label{alg:path}
\begin{algorithmic}[1]
\Require A connected graph $H=(V,E)$.
\Ensure A minimum-cost efficient dominating broadcast whose domination graph is a path, with $\PathCost(K_1)=0$ by convention.
\If{$|V|=1$}
  \State Return the zero broadcast of cost $0$.
\EndIf
\State Compute all-pairs distances and $\rho=\rad(H)$.
\State For every ball $(v,p)$, $1\le p\le\rho$, compute the components of $H[V\setminus \Ball(v,p)]$.
\State Create all states: radial states, endpoint states, and internal oriented states.
\State Compute all values $\Req(v,p,c,w)$.
\State Build the state digraph $D(H)$ by enumerating triples $(v,w,p)$ and setting $q=\dist(v,w)-p-1$.
\State Topologically order states by increasing $|L_\sigma|$.
\State Initialize $d[\sigma]=r_\sigma$ for source states and $d[\sigma]=+\infty$ otherwise.
\For{states $\sigma$ in topological order}
  \For{each arc $\sigma\to\tau$}
    \State $d[\tau]\gets \min\{d[\tau],d[\sigma]+r_\tau\}$.
  \EndFor
\EndFor
\State Return a sink state of minimum $d$-value and reconstruct the corresponding broadcast by predecessor pointers.
\end{algorithmic}
\end{algorithm}

\begin{theorem}
\label{thm:path}
Algorithm~\ref{alg:path} computes $\PathCost(H)$ and a corresponding broadcast in $O(s^3)$ time and $O(s^3)$ space.
\end{theorem}

\begin{proof}
If $s=1$, the algorithm returns the zero broadcast of cost $0$, which is correct by the convention $\PathCost(K_1)=0$. Assume $s\ge 2$.

By Lemma~\ref{lem:soundness}, every source-to-sink path of $D(H)$ gives a feasible efficient dominating broadcast whose domination graph is a path. By Lemma~\ref{lem:completeness}, every such broadcast appears as a source-to-sink path of the same weight. Thus the minimum source-to-sink path in $D(H)$ has value $\PathCost(H)$.

By Lemma~\ref{lem:dag}, $D(H)$ is acyclic and a topological order is given by increasing $|L_\sigma|$. Therefore dynamic programming computes the minimum source-to-sink path. The preceding construction builds all states and arcs in $O(s^3)$ time and uses $O(s^3)$ space for distances, component labels, and requirement values. The dynamic program takes $O(|V(D)|+|E(D)|)=O(s^3)$ time.
\end{proof}

\section{\texorpdfstring{The $O(n^5)$ general algorithm}{The O(n5) general algorithm}}
\label{sec:general}

We now combine the cubic path-case solver with the peel lemma.

For a candidate ball $(x,k)$ in the original graph $G$, let
\[
  X(x,k)=\Ball_G(x,k),
  \qquad
  G_{x,k}=G[V(G)\setminus X(x,k)].
\]
If $G_{x,k}$ is empty, the candidate cost is $k$. If $G_{x,k}$ consists of a single vertex $y$, the candidate broadcast assigns power $1$ to $y$ and has cost $k+1$. If $G_{x,k}$ is connected and has at least two vertices, run Algorithm~\ref{alg:path} on $G_{x,k}$ and add $k$ to the resulting path cost. The best candidate over all $x$ and $k$ is returned.

The explicit singleton case is necessary only in the outer algorithm. Although we set $\PathCost(K_1)=0$ as a standalone convention, a singleton residual vertex is outside the peeled ball and is not dominated by that ball. Therefore it must receive positive power when the residual solution is extended back to the original graph.

\begin{algorithm}[t]
\caption{\textsc{Fast-Optimal-Broadcast}$(G)$}
\label{alg:general}
\begin{algorithmic}[1]
\Require A connected graph $G=(V,E)$ with $|V|\ge 2$.
\Ensure An optimal broadcast domination of $G$.
\State Compute all-pairs distances in $G$ and $\rad(G)$.
\State $best\gets \rad(G)$; let $bestBroadcast$ be a radial broadcast of cost $\rad(G)$.
\For{each $x\in V$}
  \For{$k=1,2,\ldots,\rad(G)$}
    \State $X\gets \Ball_G(x,k)$ and $H\gets G[V\setminus X]$.
    \If{$H=\emptyword$}
      \If{$k<best$}
        \State $best\gets k$; $bestBroadcast\gets$ the broadcast assigning power $k$ to $x$.
      \EndIf
    \ElsIf{$|V(H)|=1$}
      \State Let $y$ be the unique vertex of $H$.
      \State $C\gets k+1$.
      \If{$C<best$}
        \State $best\gets C$.
        \State $bestBroadcast\gets f$, where $f(x)=k$, $f(y)=1$, and $f(v)=0$ for all other vertices.
      \EndIf
    \ElsIf{$H$ is connected}
      \State \Comment{Disconnected $H$ is safely skipped per Corollary~\ref{cor:peel}}
      \State Run \textsc{Fast-Path-Broadcast}$(H)$, obtaining a broadcast $g$ on $H$.
      \State $C\gets k+\sum_{v\in V(H)}\min\{g(v),\ecc_G(v)\}$.
      \If{$C<best$}
        \State $best\gets C$.
        \State $bestBroadcast\gets f$, where $f(x)=k$, $f(v)=\min\{g(v),\ecc_G(v)\}$ for $v\in V(H)$, and $f(v)=0$ otherwise.
      \EndIf
    \EndIf
  \EndFor
\EndFor
\State Return $bestBroadcast$.
\end{algorithmic}
\end{algorithm}

\begin{lemma}[Every candidate is feasible]
\label{lem:candidatefeasible}
Every broadcast considered by Algorithm~\ref{alg:general} is a dominating broadcast of $G$.
\end{lemma}

\begin{proof}
If $H=\emptyword$, the ball $\Ball_G(x,k)$ covers all of $G$.

If $|V(H)|=1$, let $y$ be the unique vertex of $H$. The algorithm assigns power $1$ to $y$. This is a legal power because $G$ is connected, $|V(G)|\ge 2$, and hence $\diam(G)\ge 1$. The vertex $y$ dominates itself, while the broadcast at $x$ dominates every vertex of $X=\Ball_G(x,k)$. Thus the resulting broadcast dominates $G$.

Otherwise, Algorithm~\ref{alg:path} returns a broadcast $g$ that dominates $H$ using distances in $H$. Since $H$ is an induced subgraph of $G$, distances in $G$ are no larger than distances in $H$. Hence the same powers dominate all vertices of $H$ in $G$. If some residual power exceeds the eccentricity of its center in $G$, capping it at $\ecc_G(v)$ preserves domination in $G$ and can only decrease the cost. The additional broadcast at $x$ dominates every vertex of $X=\Ball_G(x,k)$.
\end{proof}

We also record a small boundary fact that rules out singleton residuals in the optimal peeled instance supplied by Corollary~\ref{cor:peel}.

\begin{lemma}
\label{lem:no-singleton-peel}
Let $G$ be a connected graph with at least two vertices, and let $f$ be an efficient dominating broadcast of $G$. If $x\in V_f$ and $k=f(x)$, then $G[V(G)\setminus\Ball_G(x,k)]$ is not a one-vertex graph.
\end{lemma}

\begin{proof}
Suppose for a contradiction that $G[V(G)\setminus\Ball_G(x,k)]$ has the single vertex $y$. Since $f$ dominates $G$ and $y\notin\Ball_G(x,k)$, the vertex $y$ is dominated by some active vertex $a\ne x$. The center $a$ cannot lie in $\Ball_G(x,k)$, because then the ball of $a$ would contain $a$ and intersect the ball of $x$, contradicting efficiency. Hence $a=y$.

Because $G$ is connected and has at least two vertices, $y$ has a neighbor $z$. Since $y$ is the only vertex outside $\Ball_G(x,k)$, this neighbor $z$ lies in $\Ball_G(x,k)$. But $f(y)>0$, so $z\in\Ball_G(y,f(y))$. Thus the active balls of $x$ and $y$ intersect at $z$, again contradicting efficiency.
\end{proof}

\begin{lemma}
\label{lem:optconsidered}
Algorithm~\ref{alg:general} considers a candidate of cost $\bd(G)$.
\end{lemma}

\begin{proof}
By Corollary~\ref{cor:peel}, there is an efficient optimal broadcast $f$ and an active vertex $x\in V_f$ such that, with $k=f(x)$, the residual graph
\[
  H=G[V(G)\setminus \Ball_G(x,k)]
\]
is empty, or is connected and satisfies
\[
  \bd(H)=\PathCost(H).
\]
By Lemma~\ref{lem:no-singleton-peel}, this residual graph is not a one-vertex graph.

If $H$ is empty, Algorithm~\ref{alg:general} considers the ball $(x,k)$ and obtains cost
\[
  k=\cost(f)=\bd(G).
\]

Otherwise, $H$ is connected and has at least two vertices. Let $f_H$ be the restriction of $f$ to $V(H)$.
Since $f$ is efficient, $f_H$ dominates $H$: indeed, any shortest path witnessing domination in $G$ by a remaining active vertex cannot enter $\Ball_G(x,k)$, or else two active balls of $f$ would intersect. Hence
\[
  \PathCost(H)=\bd(H)\le \cost(f_H)=\cost(f)-k.
\]
Therefore
\[
  k+\PathCost(H)\le \cost(f).
\]

If this inequality were strict, then combining the broadcast at $x$ with an optimal path broadcast of $H$ gives, after capping residual powers at their eccentricities in $G$ if necessary, a dominating broadcast of $G$ with cost strictly smaller than
\[
  \cost(f)=\bd(G),
\]
contradicting optimality of $f$. Hence
\[
  k+\PathCost(H)=\cost(f)=\bd(G).
\]

The iteration for $(x,k)$ therefore obtains a candidate of cost at most $\bd(G)$. Since every candidate is feasible by Lemma~\ref{lem:candidatefeasible}, no candidate can have cost below $\bd(G)$. Hence this iteration considers a candidate of cost exactly $\bd(G)$.
\end{proof}
\begin{theorem}
\label{thm:general}
Algorithm~\ref{alg:general} computes an optimal broadcast domination of every connected unweighted graph $G$ on $n\ge 2$ vertices in $O(n^5)$ time and $O(n^3)$ working space. Together with the convention $\bd(K_1)=0$, this covers all connected unweighted graphs.
\end{theorem}

\begin{proof}
By Lemma~\ref{lem:candidatefeasible}, the algorithm never returns a value smaller than $\bd(G)$. By Lemma~\ref{lem:optconsidered}, it considers a candidate of value exactly $\bd(G)$. Thus it is correct.

There are $O(n^2)$ candidate balls $(x,k)$, because $x\in V(G)$ and $1\le k\le \rad(G)\le n-1$. Connectivity of each residual graph can be tested in $O(n^2)$ time, for a total of $O(n^4)$, or precomputed within the same bound. Empty residuals and singleton residuals are handled directly. For every connected residual graph $H$ with at least two vertices, Algorithm~\ref{alg:path} takes $O(|V(H)|^3)\le O(n^3)$ time. Therefore the total running time is
\[
  O(n^2)\cdot O(n^3)=O(n^5).
\]
The path routine uses $O(n^3)$ working space, and the outer loop reuses this space.
\end{proof}

\section{Implementation notes}

The algorithm above is most naturally implemented with integer labels for components and sides.

For every ball $(v,p)$, store:
\begin{itemize}[leftmargin=2em]
  \item the number of residual components $\kappa(v,p)$;
  \item for every vertex $z\notin \Ball(v,p)$, a local component index $\Comp(v,p,z) \in \{1, 2\}$;
  \item the size of each residual component, used to compute $|L_\sigma|$.
\end{itemize}

A state can then be represented by four integers:
\[
  (v,p,\ell,r),
\]
where $\ell, r \in \{0, 1, 2\}$ are local component indices, and index $0$ denotes the empty side. The source states are those with $\ell=0$; the sink states are those with $r=0$.

To maintain the $O(n^3)$ space bound, the four-dimensional requirement array $\Req$ is indexed by these local component identifiers rather than global component labels, yielding an array of size $n \times n \times 2 \times n$. 

When enumerating arcs from $(v,p,\ell,r)$ to $(w,q,\ell',r')$, the test is constant-time:
\begin{enumerate}[label=(\roman*),leftmargin=3em]
  \item $r\ne 0$ and $\ell'\ne 0$;
  \item $q=\dist(v,w)-p-1$ and $1\le q\le \rho$;
  \item $\Comp(v,p,w)=r$ and $\Comp(w,q,v)=\ell'$;
  \item $\Req(v,p,r,w)\le q$ and $\Req(w,q,\ell',v)\le p$.
\end{enumerate}

The dynamic program stores predecessor pointers to reconstruct the path of states. The final broadcast assigns to every center appearing on that path the radius stored in its state.

In the general algorithm, singleton residual graphs should not be passed to the path routine and interpreted with cost $0$. Instead, if deleting $\Ball_G(x,k)$ leaves the single vertex $y$, the implementation should evaluate the explicit feasible candidate with $f(x)=k$ and $f(y)=1$. This preserves the standalone convention $\PathCost(K_1)=0$ while keeping every outer-loop candidate a valid broadcast of the original graph.

\section{Experimental evaluation}
\label{sec:experiments}

This section reports a small reference-implementation experiment comparing the algorithm of Section~\ref{sec:general} with an implementation of the previous Heggernes--Lokshtanov framework. The purpose of the experiment is not to provide hardware-normalized performance claims, but to check that the theoretical improvement in the path-case routine is visible in an executable implementation.

\subsection{Implementations and methodology}

Both implementations were written in Python using the same graph representation, the same all-pairs distance routine, the same integer-bitset operations for balls and residual graphs, and the same safe incumbent pruning. The only intended difference between the two solvers was the path-case subroutine.

The baseline solver follows the previous state-of-the-art structure of Heggernes and Lokshtanov~\cite{HL06}: it uses the same peel-one-ball outer loop, but invokes an anchored path-case routine in the style of their RMPBD algorithm. The new solver uses Algorithm~\ref{alg:path} as the path-case routine. Thus the experiment isolates the change from the anchored path solver to the anchor-free state-DAG solver as much as possible within a Python reference implementation.

For every test instance, both solvers were run on the same input graph and their returned costs were compared. All instances in this experiment returned identical optimum costs. Each table entry reports median wall-clock time over the runs recorded by the benchmark script. The tests were run under Python~3.13.5 in a Linux container with approximately 4GB of memory; the containerized environment is useful for reproducibility but should not be interpreted as a calibrated benchmarking platform.

\subsection{Benchmark families}

The benchmark suite contained seven graph families: paths, cycles, random trees, sparse connected Erd\H{o}s--R\'enyi graphs, barbells, stars, and wheels. These families were chosen to separate two effects. Paths and cycles force substantial use of the path-case machinery, so they are expected to show the clearest improvement. Stars and wheels usually have immediate radial or near-radial optima, so both algorithms terminate quickly and little speedup is expected.

\begin{table}[t]
\centering
\small
\begin{tabular}{lrrrrrr}
\toprule
Family & Cases & Max $n$ & Median speedup & Max speedup & New time & Baseline time \\
\midrule
Barbell & 5 & 28 & $3.06\times$ & $4.25\times$ & 0.0703s & 0.2992s \\
Cycle & 7 & 32 & $3.93\times$ & $6.66\times$ & 0.7185s & 4.7873s \\
Path & 9 & 40 & $4.77\times$ & $8.39\times$ & 0.6902s & 5.7876s \\
Random tree & 6 & 30 & $3.05\times$ & $6.28\times$ & 0.0118s & 0.0743s \\
Sparse random & 6 & 30 & $2.30\times$ & $5.12\times$ & 0.0313s & 0.1603s \\
Star & 6 & 160 & $1.02\times$ & $1.26\times$ & 0.0048s & 0.0049s \\
Wheel & 9 & 80 & $0.99\times$ & $1.47\times$ & 0.0019s & 0.0019s \\
\bottomrule
\end{tabular}
\caption{Timing comparison between the anchor-free implementation and an HL-style anchored baseline. ``New time'' and ``Baseline time'' are the median running times on the largest compared instance in the family. Speedup is baseline time divided by new time.}
\label{tab:experiments-summary}
\end{table}

\begin{figure}[t]
\centering
\IfFileExists{benchmark_against_sota_speedup.png}{%
\includegraphics[width=.82\linewidth]{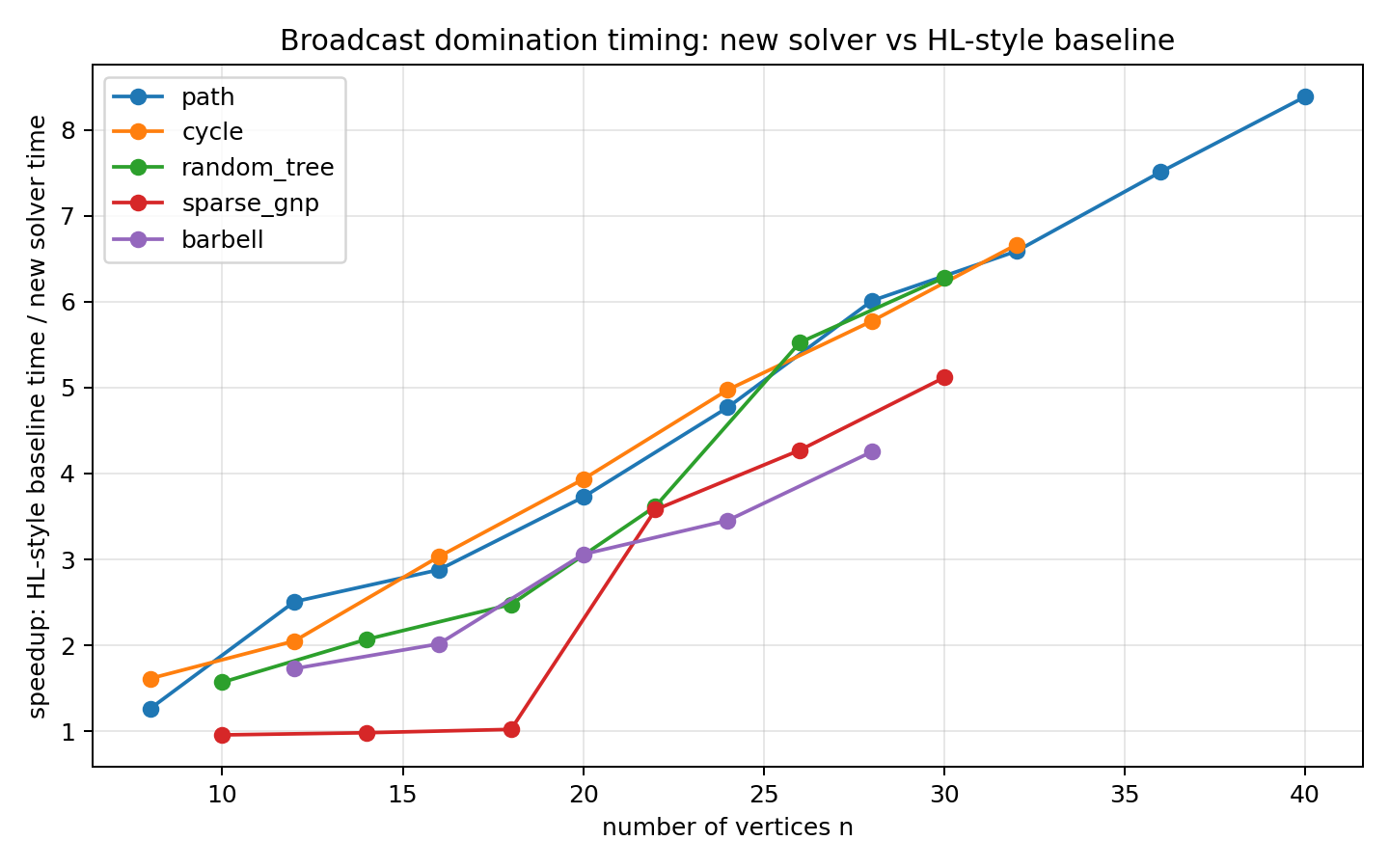}%
}{%
\fbox{\begin{minipage}{.78\linewidth}\centering
The file \texttt{benchmark\_against\_sota\_speedup.png} was not found.
\end{minipage}}%
}
\caption{Measured speedup of the anchor-free implementation over the HL-style anchored baseline. The largest improvements occur on graph families where the path-case subroutine dominates the running time.}
\label{fig:experiments-speedup}
\end{figure}

\subsection{Results}

Table~\ref{tab:experiments-summary} summarizes the comparison, and Figure~\ref{fig:experiments-speedup} plots the same speedup data. On paths and cycles, the speedup grows steadily with $n$. At the largest path instance tested, $n=40$, the new implementation ran in 0.6902 seconds while the anchored baseline took 5.7876 seconds, a speedup of $8.39\times$. At the largest cycle instance tested, $n=32$, the new implementation ran in 0.7185 seconds while the baseline took 4.7873 seconds, a speedup of $6.66\times$.

The improvement is also visible on random trees, sparse connected random graphs, and barbells, with median speedups between $2.30\times$ and $3.06\times$. On stars and wheels, the two implementations are essentially tied. This is consistent with the structure of these instances: their optima are radial or near-radial, and both implementations identify good candidates before spending much time inside the path-case routine.

\subsection{Interpretation}

The experiment supports the theoretical analysis in two ways. First, the returned objective values agreed on every benchmark instance, providing an additional implementation check beyond the exhaustive and randomized exact-oracle tests used during development. Second, the speedups are largest exactly where the analysis predicts they should be largest: on instances where the algorithm repeatedly invokes the path-case solver on large residual graphs. The modest speedup on easy radial families is not a contradiction of the asymptotic improvement; it reflects the fact that the path-case bottleneck is absent or nearly absent on those inputs.

The benchmark should be read as a reference-implementation comparison. The baseline is an HL-style implementation rather than the original authors' code, and both solvers are written in Python. An optimized implementation in C or C++ could change constant factors substantially. Nevertheless, because the two solvers share the same infrastructure and differ primarily in the path-case subroutine, the observed speedups are consistent with the theoretical replacement of an $O(n^4)$ anchored path routine by the $O(n^3)$ anchor-free routine of Algorithm~\ref{alg:path}.

\section{Concluding remarks}

The improvement from $O(n^6)$ to $O(n^5)$ is obtained without changing the global path/cycle structure theorem. The only change is in the path-case subroutine: instead of choosing an anchor vertex and solving $n$ related acyclic shortest-path problems, the algorithm orients every candidate separator ball directly and solves one acyclic shortest-path problem. The peel-one-ball reduction of Heggernes and Lokshtanov then immediately yields the general $O(n^5)$ bound, thereby resolving the quintic-time conjecture for general graphs attributed to Heggernes and S\ae ther.

A further improvement to $O(n^4)$ would require avoiding the $O(n^2)$ peel-ball loop or solving all residual path instances simultaneously. The state graph developed here is acyclic in the path case, but the analogous construction for cycles becomes a genuine compatible-cycle problem and loses the monotonicity used in Lemma~\ref{lem:dag}.

\end{document}